\documentclass{amsart}

\usepackage{amssymb}
\usepackage{amsthm}
\usepackage{amsmath,amsfonts,amsthm}
\usepackage{thmtools}
\usepackage{etoolbox}
\usepackage{mathtools}
\usepackage{soul}
\usepackage{tcolorbox}
\usepackage{times}
\usepackage{algorithm}
\usepackage{algorithmicx}
\usepackage{algpseudocode}
\usepackage{tabularx}
\usepackage{caption}
\usepackage{latexsym} 
\usepackage{fleqn}
\usepackage{tikz}
\usepackage{subfiles}
\usepackage{microtype}
\usepackage{url}

\usepackage[a4paper, total={6in, 8in}]{geometry}

\newtheorem{definition}{Definition}

\newtheorem{lemma}[definition]{Lemma}
\newtheorem{theorem}[definition]{Theorem}

\newtheorem{example}[definition]{Example}

\makeatletter
\newcommand{\multiline}[1]{%
  \begin{tabularx}{\dimexpr\linewidth-\ALG@thistlm}[t]{@{}X@{}}
    #1
  \end{tabularx}
}
\makeatother

\newcommand{\ar}{\ensuremath{\mathrm{ar}}}

\newcommand{\problemDef}[3]
{%
    
    \begin{tcolorbox}[arc=0.1mm,boxsep=-0.6mm,left=1.9mm,right=1.9mm,bottom=1.4mm,top=1.4mm,adjusted title={\strut \sc#1},colback=white!5]

    \noindent\textbf{Instance:} #2
    
    \noindent\textbf{Question:} #3
    \end{tcolorbox}
}

  \newcommand{\multiproblemDef}[4]
{%
    
    \begin{tcolorbox}[arc=0.1mm,boxsep=-0.6mm,left=1.9mm,right=1.9mm,bottom=1.4mm,top=1.4mm,adjusted title={\strut \sc#1},colback=white!5]

    \noindent\textbf{Instance:} #2

    \noindent\textbf{Parameters:} #3
    
    \noindent\textbf{Question:} #4
    \end{tcolorbox}
  }

\newcommand{\Ordo}{{\mathcal{O}}}

\newcommand{\para}[1]{\ensuremath{\text{para-}#1}}

\newcommand{\x}[1]{\ensuremath{\mathbf{X}#1}}

\newcommand{\E}{\ensuremath{\mathbf{E}}}
\newcommand{\FPE}{\ensuremath{\mathbf{FPE}}}
\newcommand{\FPT}{\ensuremath{\mathbf{FPT}}}
\newcommand{\XP}{\ensuremath{\mathbf{XP}}}

\newcommand{\parared}{\ensuremath{\leq^{\FPE}}}

\newcommand{\ered}{\ensuremath{\leq^{\E}}}

\newcommand*{\ew}{effective width}
\newcommand*{\Ew}{Effective width}

\author{Leif Eriksson}
\address[L. Eriksson]%
   {Dep. Computer and Information
     Science, \\Link\"opings  Universitet, Sweden
   }
\email{leif.eriksson@liu.se}   
\author{Victor Lagerkvist}
\address[V. Lagerkvist]%
   {Dep. Computer and Information
     Science, \\Link\"opings  Universitet, Sweden}
\email{victor.lagerkvist@liu.se}      

\title{A Multivariate Complexity Analysis of Qualitative Reasoning Problems}

\begin{document}
\maketitle

\begin{abstract}
Qualitative reasoning is an important subfield of artificial intelligence where one describes relationships with qualitative, rather than numerical, relations. Many such reasoning tasks, e.g., \textsc{Allen's interval algebra}, can be solved in $2^{\Ordo(n \cdot \log n)}$ time, but single-exponential running times $2^{\Ordo(n)}$ are currently far out of reach. 
In this paper we consider single-exponential algorithms via a {\em multivariate} analysis consisting of a {\em fine-grained} parameter $n$ (e.g., the number of variables) and a {\em coarse-grained} parameter $k$ expected to be relatively small.
We introduce the classes $\FPE$ and $\x{\E}$ of problems solvable in $f(k) \cdot 2^{\Ordo(n)}$, respectively $f(k)^n$, time, and prove several fundamental properties of these classes. We proceed by studying temporal reasoning problems and (1) show that the \textsc{partially ordered time} problem of effective width $k$ is solvable in $16^{kn}$ time and is thus included in $\x{\E}$, and (2) that the network consistency problem for \textsc{Allen's interval algebra} with no interval overlapping with more than $k$ others is solvable in $(2nk)^{2k} \cdot 2^{n}$ time and is included in $\FPE$. Our multivariate approach is in no way limited to these to specific problems and may be a generally useful approach for obtaining single-exponential algorithms.
\end{abstract}

\section{Introduction}
{\em Qualitative reasoning} is an important formalism in artificial intelligence where the objective is to reason about continuous properties given certain relations between the unknown entities, expressed qualitatively rather than quantitatively (e.g., numerically). Two important subfields are {\em spatial} reasoning, where basic objects e.g.\ may be regions in space defined in terms of other regions, and {\em temporal} reasoning, where one wishes to describe the relationship between time points, and more generally, between intervals of time points. Examples of the former include the \textsc{region connection calculus} and the \textsc{cardinal direction calculus}, and to exemplify the latter we may e.g. mention \textsc{Allen’s interval algebra}, the \textsc{branching time} problem,  and the \textsc{partially ordered time} problem. For many more examples and practical applications of these formalisms, cf.\ the survey by Dylla et al.~\cite{Dylla:2017:SQS:3058791.3038927}.
The classical complexity of qualitative reasoning is well understood: they are generally NP-hard but form non-trivial tractable fragments by restricting the types of basic relations. However, the tractable cases have limited strength, and in practice we therefore also need to solve NP-hard reasoning tasks. Hence, is there any hope of solving these intractable problems, or do we have to be contempt with heuristics?
Here, modern complexity theory tells a more nuanced story than the classical theory, where  NP-hardness is the beginning rather than the end. There are two prominent views: 

\begin{enumerate}
\item
Fine-grained complexity: for a {\em fine-grained} parameter $n$, e.g., the number of variables, how fast can the problem be solved, and are existing algorithms close to being optimal (given complexity theoretical assumptions such as {\em the (strong) exponential-time hypothesis} ((S)ETH))?
\item
Parameterized complexity: for a {\em coarse-grained} parameter $k$, e.g., tree-width, which problems are {\em fixed-parameter tractable} ($\FPT$) --- are in P if one is allowed unlimited computational resources (with respect to $k$) to preprocess the instance?
\end{enumerate}

Crucially,  $n$ is expected to grow with the size of the instance while $k$ is expected to stay  small. Hence, fine-grained and parameterized complexity are not competing methods but tackle different aspects of intractability, with diverging algorithmic toolboxes.
Unfortunately, neither approach seem fit for qualitative reasoning problems. On the one hand, they are solvable in $2^{\Ordo(n^2)}$ time or $2^{\Ordo(n \cdot \log n)}$ time in certain cases~\cite{lagerkvist2017d}, but we can currently only rule out subexponential $2^{o(n)}$ algorithms under  ETH~\cite{jonsson2021}. On the other hand, despite the immense success of parameterized complexity, there is a lack of natural $\FPT$ algorithms for qualitative reasoning, and we are only aware of a handful of less surprising examples such as tree-width~\cite{10.1016/j.jcss.2012.05.012,DBLP:conf/aaai/DabrowskiJOO21}.

Hence, are we asking the right questions by attacking these problems with parameterized and fine-grained complexity, or are current attempts misguided?
Could it be the case that they simply are too hard to admit natural $\FPT$ algorithms? Similarly, could it be the case that (e.g.) \textsc{Allen's interval algebra} is not solvable in single-exponential time while still being  fundamentally too different from 3-SAT to admit stronger lower bounds than $2^{o(n)}$? In this article we bridge the gap between fine-grained and parameterized complexity and attack this question by a {\em multivariate} complexity analysis of single-exponential time solvability: which NP-hard qualitative reasoning problems admit single-exponential time algorithms (with respect to a fine-grained parameter $n$) once preprocessed with respect to a coarse-grained parameter $k$? This approach is a natural continuation of Bringmann \& K\"unnemann~\cite{bringmann2018} who studied the complexity of the \textsc{longest common subsequence} problem in order to investigate which natural parameters could possibly break the SETH barrier. However, the exponential world of NP-hard qualitative reasoning problems is very different from the tractable \textsc{longest common subsequence} problem and requires new tools and techniques.
Our paper has two main contributions. First, in Section~\ref{sec:classes} we initiate a systematic investigation of complexity classes taking both a coarse-grained parameter $k$ and a fine-grained parameter $n$ into consideration. We identify two natural classes: $\FPE$, problems solvable in $f(k) \cdot 2^{\Ordo(n)}$ time, and $\x{\E}$, problems solvable in $f(k)^n$ time. Naturally, $\FPE$ is more desirable due to the complete separation of the parameter $n$ and $k$ and should be viewed as the exponential analogue of $\FPT$, while $\x{\E}$ corresponds to the less desirable class $\XP$ where the two parameters are intertwined. 
Second, in Section~\ref{sec:algorithms} we begin the multivariate analysis of single-exponential time solvability by classifying parameterized problems as belonging to $\FPE$, $\x{\E}$, or neither (under the ETH). We first consider the finite-domain {\em constraint satisfaction problem} (CSP) and manage to rule out inclusion in both $\FPE$ and $\x{\E}$ under the ETH for several natural parameterizations. We then turn to two well-known problems from temporal reasoning: the \textsc{partially ordered time} problem~\cite{Anger:etal:flairs99}, which is the problem of determining whether there exists a partial order subjected to a set of constraints over a variable set, and \textsc{Allen's interval algebra}~\cite{allen83}, the problem of determining whether a set of 2-dimensional intervals described by 1-dimensional temporal constraints over the  start- and end-points, is consistent or not. \textsc{Allen's interval algebra} has seen applications in e.g.\ planning~\cite{DBLP:conf/ijcai/AllenK83,DBLP:journals/dke/Dorn95,DBLP:conf/icra/MudrovaH15}, natural language processing~\cite{DBLP:conf/ijcai/DenisM11,DBLP:conf/aaai/SongC88}, and molecular biology~\cite{DBLP:journals/jacm/GolumbicS93}, while the \textsc{partially ordered time} problem occurs naturally in distributed systems where time is not totally ordered, cf.\ Lamport's classical interprocessor communication model~\cite{10.1145/5383.5384,10.1145/65979.65982}. Both of these problems can be solved in $2^{\Ordo(n \cdot \log n)}$ time by enumeration, and \textsc{Allen's interval algebra} is additionally known to admit an improved $\Ordo(1.0615n^{n})$ algorithm~\cite{lagerkvist2021b}.  For \textsc{partially ordered time} we consider {\em \ew} as parameter, which in strength lies between the well-known properties of {\em width} and {\em dimension}. Here, we construct an $\x{\E}$-algorithm with a running time dominated by $16^{kn}$ where $k$ is the {\ew} of the partial order. The algorithm is non-trivial and attempts to construct a partial order of {\ew} $k$ by guessing a suitable partition of the variable set and the ordering among all variables, which can be efficiently implemented via a recursive strategy making use of dynamic programming. For \textsc{Allen's interval algebra} we define a parameter $k$ based on the maximum number of overlapping intervals and construct a $(2nk)^{2k} \cdot 2^{n}$ time algorithm, and thus prove membership in $\FPE$. Hence, for instances where $k$ stays reasonably small our algorithm is effectively as good as a $2^{n}$-time algorithm, which is a major improvement over the aforementioned $\Ordo(1.0615n^{n})$ time algorithm. It is worth mentioning that both of these algorithms also solve the more involved problem of {\em counting} the number of solutions. 

Our work opens up several directions for future research, which we describe in greater detail in Section~\ref{sec:discussion}. Most importantly, which parameterized problems in qualitative reasoning, and artificial intelligence as a whole, belong to $\FPE$, and which belong to $\x{\E}$? 

\section{Preliminaries}

Throughout, let $\Sigma$ be an alphabet, i.e., a finite set of symbols. For $x \in \Sigma^*$ we let $|x|$ be its length.

\begin{definition}
  A {\em parameterized problem} is a subset of $\Sigma^* \times \mathbb{N}$ where $k \in \mathbb{N}$ is called the {\em parameter}.
\end{definition}  

The goal of parameterized complexity is then to describe the complexity in terms of the parameter, and, ideally,  design algorithms which decouples the parameter to  $|x|$.

\begin{definition}
  We introduce the following running times and the corresponding classes from parameterized complexity:
  \begin{enumerate}
  \item
    $f(k) \cdot |x|^{\Ordo(1)}$ time by a deterministic algorithm ($\FPT$),
  \item
    $|x|^{f(k)}$ time by a deterministic algorithm ($\XP$),
  \end{enumerate}
  where $f \colon \mathbb{N} \rightarrow \mathbb{N}$ is a computable function. 
\end{definition}

$\FPT$ can be seen as the parameterized version of P where the parameter is completely decoupled from the rest of the instance via the factor $f(k)$. The class $\XP$, in comparison, still yields a polynomial-time algorithm for every fixed $k$ but the dependency on $k$ is much worse.

If a {\em fine-grained} parameter such as the number of variable $n$ is used as parameter then the objective is typically to obtain precise upper and lower bounds on running times of the form $f(n) \cdot |x|^{\Ordo(1)}$, and the resulting field is sometimes called {\em fine-grained complexity}. If we let $\mathbf{SE}$ be the subclass of $\FPT$ allowing subexponential running times of the form \mbox{$2^{o(n)} \cdot |x|^{\Ordo(1)}$} then the conjecture that 3-SAT is not included in $\mathbf{SE}$ is known as the {\em exponential-time hypothesis} (ETH). 
The ETH is widely believed conjecture and is a powerful assumption for proving conditional lower bounds. However, sometimes stronger assumptions are necessary. For $k \geq 3$ let $c_k$ be the infimum of all constants $c$ such that $k$-\textsc{SAT} is solvable in $c^n \cdot |x|^{\Ordo(1)}$ time. The conjecture that the limit of the sequence $c_3, c_4, \ldots$ tends to $2$ is known as the {\em strong exponential-time hypothesis} (SETH). Importantly, the SETH implies that \textsc{CNF-SAT}, the satisfiability problem where we do not have any restrictions on clause lengths, is not solvable in $c^n \cdot |x|^{\Ordo(1)}$ time for {\em any} $c < 2$. 

Before introducing the qualitative reasoning problems under consideration in this paper we introduce the more general class of {\em constraint satisfaction problems} (CSPs).

\problemDef{CSP}
{$I = (V,C)$ where $V$ is a set of variables and $C$ a set of constraints over a domain $D$.}
{Does there exist a function $f:V \rightarrow D$ which satisfies all constraints, i.e., $(f(x_1), \ldots, f(x_m)) \in R$ for every $R(x_1, \ldots, x_m) \in C$?}

A set of relations $\Gamma$ naturally induces a problem CSP$(\Gamma)$ where constraints only uses relations from $\Gamma$.
For finite-domain constraint languages $\Gamma$ we represent relations by explicitly listing the tuples in the relation, but for infinite domains we represent relations by first-order formulas over a fixed set of base relations $\mathcal{B}$, i.e., each relation is defined as the set of models of a first-order formula, interpreted over $\mathcal{B}$. Additionally, if $\mathcal{B}$ is a set of relations of the same arity then we
\begin{enumerate}
  \item
    write $\mathcal{B}^{\vee k}$ for the set of all relations definable by disjunctions of arity $k$ over positive atoms from $\mathcal{B}$, i.e., an $n$-ary $R \in \mathcal{B}^{\vee k}$ if and only if $R(x_1, \ldots, x_n) \equiv R_1(\mathbf{x}_1) \lor \ldots \lor R_m(\mathbf{x}_m)$ where each $R_i \in \mathcal{B}$ and each $\mathbf{x}_i$ is a tuple of variables matching the arity of $R_i$, and
    \item
 write $\mathcal{B}^{\vee =}$ for the set of all relations that can be obtained by unions of the relations in $\mathcal{B}$, which via the logical perspective just corresponds to definitions $R(x_1, \ldots, x_n) \equiv R_1(\mathbf{x}) \lor \ldots \lor R_m(\mathbf{x})$ where  $\mathbf{x}$ is a tuple of variables matching the arity of $R_1, \ldots, R_m$. 
\end{enumerate}

\begin{definition}\label{partialorder}
A pair $(S, \leq)$ is a \emph{partial order} if $\leq$ is reflexive, asymmetric ($\forall x, y \in S$, $x \leq y \Rightarrow $ not $y \leq x$), and transitive.
\end{definition}

If $\odot \in \{<,>,||,=\}$ and $P = (S, \leq_P)$ is a partial order then we write $\odot_P$ for the relation induced by $P$: $x <_P y$ if $x \leq_P y$ and $y \leq_P x$ does not hold, conversely for $>_P$, $||_P$ if neither $x \leq_P y$ nor $y \leq_P x$, and $x =_P y$ if $x \leq_P y$ and $y \leq_P x$.
We will differ between relations in $P$ and the relations induced by $P$ in the same manner, i.e. $x\leq_Py$ is a relation in $P$, while $x<_Py$ is a relation induced by $P$.

\problemDef{Partially Ordered Time}
{A set of variables $V$ and a set of binary constraints $C$ where $c \subseteq \{<,>,||,=\}$ for each $c(x,y) \in C$.}
{Is there a partial order $P = (S, \leq)$ with ${|S|\leq |V|}$ and a function $f \colon V\rightarrow S$ such that for every constraint $c(x,y)\in C$, $f(x) \odot_P f(y)$ for some $\odot \in c$?}

For a \textsc{partially ordered time} instance $I = (V,C)$ and $X \subseteq V$ we write $I[X]=(X,C_X)$ for the sub-instance where $C_X = \{c(x,y) \in C \mid x,y \in X\}$. Similarly, if $P = (S, \leq)$ is a partial order and $S' \subseteq S$ then we by $P[S']$ denote the partial order induced by $P$ and $S'$.

The second major problem analysed in the paper is \textsc{Allen's interval algebra}:
define $a$ as the thirteen possible atomic relations between two intervals on the same line (see Table~\ref{tab:IAdesc} for precise definitions of these relations). The consistency problem for \textsc{Allen's interval algebra} can then be seen as CSP$(\mathcal{A})$ where $\mathcal{A} = a^{\vee =}$, i.e., each constraint can be expressed as a union of basic constraints.
We represent intervals as pairs of start- and end-points $X=(x^-,x^+)$ such that \mbox{$x^-<x^+$}.
In the same manner, for any CSP($\mathcal{A}$) $(V,C)$ we introduce the two auxiliary sets \mbox{$V^-=\{x^-\,|\,(x^-,x^+)\in V\}$} and \mbox{$V^+=\{x^+\,|\,(x^-,x^+)\in V\}$}. 

\begin{table}[h]
    \caption[]{The 13 basic relations between two intervals on the same line. ($i$ denotes the $i$nverse/converse of a relation.)}
    \centering
    \begin{tabular}{l|l|l}
         $\boldsymbol{Relation}$ & $\boldsymbol{Illustration}$ & $\boldsymbol{Interpretation}$ \\
         \hline
         $\boldsymbol{X<Y}$ & XXX& X precedes Y \\
         $\boldsymbol{Y>X}$ & $\,\,\,\,\,\,\,\,\,\,\,\,\,\,\,\,\,\,$YYY & \\
         \hline
         $\boldsymbol{X\,=\,Y}$ & XXXXXXX & X is equal to Y \\
          & YYYYYYY &  \\
         \hline
         $\boldsymbol{X\,m\,Y}$ & XXX & X meets Y \\
         $\boldsymbol{Y\,mi\,X}$ & $\,\,\,\,\,\,\,\,\,\,\,\,\,$YYY & \\
         \hline
         $\boldsymbol{X\,o\,Y}$ & XXXXX & X overlaps with Y \\
         $\boldsymbol{Y\,oi\,X}$ & $\,\,\,\,\,\,\,\,\,\,\,\,\,\,$YYYY & \\
         \hline
         $\boldsymbol{X\,s\,Y}$ & XXX & X starts Y \\
         $\boldsymbol{Y\,si\,X}$ & YYYYYYY & \\
         \hline
         $\boldsymbol{X\,d\,Y}$ & $\,\,\,\,\,\,\,\,\,$XXX & X during Y \\
         $\boldsymbol{Y\,di\,X}$ & YYYYYYY & \\
         \hline
         $\boldsymbol{X\,f\,Y}$ & $\,\,\,\,\,\,\,\,\,\,\,\,\,\,\,\,\,\,$XXX & X finishes Y \\
         $\boldsymbol{Y\,fi\,X}$ & YYYYYYY & \\
    \end{tabular}
    \label{tab:IAdesc}
\end{table}

\section{Multivariate Complexity Classes}
\label{sec:classes}

Our main interest are problems solvable in $2^{f(n)}$ time for $f \in \Ordo(n)$ (henceforth written $2^{\Ordo(n)}$). This naturally leads to the following parameterized variant of the classical complexity class $\E$, which is typically defined via the parameter $|x|$.

\begin{definition}
$\E$ is the class of parameterized problems solvable in $2^{\Ordo(n)} \cdot |x|^{\Ordo(1)}$ time.
\end{definition}

The choice of parameter $n$ greatly influences the existence of a $2^{\Ordo(n)}$ algorithm. For example,  if we use the number of constraints $m$ as the parameter then CSP$(\mathcal{A})$ is in $\E$ via a trivial backtracking algorithm, but it is only known to be solvable in $2^{\Ordo(n \cdot \log n)}$ time. Similarly, a $c^n \cdot |x|^{\Ordo(1)}$ algorithm for $c < 2$ for \textsc{CNF-SAT} would constitute a major break through in complexity theory, but if we instead use the number of clauses $m$ as the complexity parameter then the problem can be solved in $\Ordo(1.2226^m)$ time~\cite{CHU202151}.

\begin{example} \label{ex:e}
  The number of problems admitting non-trivial single-exponential time algorithms is vast and we only consider a handful of illustrative examples.
  \begin{itemize}
    \item
The \textsc{travelling salesman} problem is trivially solvable in $n! \cdot |x|^{\Ordo(1)}$ time but can be solved in $2^n \cdot |x|^{\Ordo(1)}$ time by the Held-Karp algorithm. Similarly, \textsc{graph coloring} can be solved in $k^n \cdot |x|^{\Ordo(1)}$ time, where $k$ is the number of colors, or $n^n \cdot |x|^{\Ordo(1)}$ time if analyzed only by the number of vertices $n$, but can be solved in $2^{n} \cdot |x|^{\Ordo(1)}$ time via an inclusion-exclusion based algorithm~\cite{DBLP:journals/siamcomp/BjorklundHK09}.
\item
  For each $k \geq 1$ let \textsc{$k$-coloring} be the problem of deciding whether an input graph can be colored with $k$ colors. Similarly, let $\textsc{coloring}$ be the variant of $\textsc{$k$-coloring}$ where $k$ is part of the input, i.e., an instance consists of a graph and a natural number $k$, and the question is whether the graph can be colored with $k$ colors. If we parameterize \textsc{$k$-coloring} by the number of vertices $n$, then, trivially, \textsc{$k$-coloring} is in $\E$ for each $k \geq 1$ since it is solvable in $k^n \cdot |x|^{\Ordo(1)}$ time, but this bound can be significantly strengthened to $2^{n} \cdot |x|^{\Ordo(1)}$  time~\cite{DBLP:journals/siamcomp/BjorklundHK09} which implies that $\textsc{coloring} \in \E$. 
 \item
 There exists many examples from parameterized algorithms where the goal is to take an $\FPT$-algorithm $f(k) \cdot |x|^{\Ordo(1)}$ and improve $f$ to a single-exponential function $2^{c \cdot k}$ where $c$ is as small as possible. Non-trivial examples include \textsc{$k$-vertex cover} which can be solved in $1.2738^k \cdot |x|^{\Ordo(1)}$ time~\cite{CHEN20103736} and the \textsc{Steiner tree problem} where $k$ bounds the set of terminal vertices in the input graph, which can be solved in $7.97^k \cdot |x|^{\Ordo(1)}$ time~\cite{DBLP:journals/siamdm/FominKLPS19}. 
 \item
   Define SNP as the class of problems expressible via second-order formulas of the form

   \[\exists R_1, \ldots, R_m, \forall x_1, \ldots, x_n \colon \varphi(R_1, \ldots, R_m, x_1, \ldots, x_n)\]

   where $\varphi$ is a quantifier-free formula with atoms from $R_1, \ldots, R_m$ and an input relational structure $\mathcal{R}$. It is then known that {\em every} problem in SNP is included in $\E$ with parameter $\Sigma^m_{i = 1} n^{\alpha_i}$ where $\alpha_i$ is the arity of $R_i$~\cite[Theorem 2]{IMPAGLIAZZO2001512}. 
\item    
    For infinite-domain examples, let $\Gamma^k$ be the set of all $k$-ary first-order definable relations over $(\mathbb{Q}; <)$ ({\em temporal} CSPs). Then CSP$(\Gamma^3)$ is solvable in $3^n \cdot |x|^{\Ordo(1)}$ time~\cite{lagerkvist2021b} but CSP$(\Gamma^4)$ is not in $\E$ under the ETH~\cite{DBLP:conf/mfcs/JonssonL18}. In contrast, if $\Delta^k$ is the set of all $k$-ary first-order definable relations over $(\mathbb{N}; =)$ ({\em equality} CSPs) then CSP$(\Delta^k)$ is solvable in  $(\frac{k(k-1)}{2})^n \cdot |x|^{\Ordo(1)}$ time~\cite{DBLP:conf/ijcai/JonssonL20} and is thus in $\E$.
  \item CSP$(\mathcal{B}^{\vee \omega})$, where $\mathcal{B}$ is a set of binary, jointly exhaustive and pairwise disjoint relations containing the equality relation over a countably infinite domain, and where
    $\mathcal{B}^{\vee \omega} = \bigcup_{k \geq 1} \mathcal{B}^{\vee k}$, is not in $\E$ under the SETH~\cite{lagerkvist2017d}.
\end{itemize}    
\end{example}

Importantly, no NP-hard qualitative reasoning tasks are known to be solvable in single-exponential time, necessitating a multivariate analysis where the complexity is measured with respect to several parameters.

\subsection{Multi-Parameterized Problems}

We begin by defining a multi-parameterized problem following the style frequently used in classical parameterized complexity.

\begin{definition}
  A {\em multi-parameterized problem} is a subset of $\Sigma^* \times \mathbb{N}^{c}$ where $n,k_1,\dots,k_{c-1} \in \mathbb{N}$ are called the {\em parameters}.
  A {\em bi-parameterized problem} is a multi-parameterized problem with only two parameters: $n$ and $k$.
\end{definition}

For simplicity we, given an instance $(x,n,k)$ of a bi-parameterized problem, will always view the first parameter $n$ as the fine-grained parameter, and the second parameter $k$ as the coarse-grained parameter.
Inspired by $\FPT$ and $\XP$ we then define the following bi-parameterized analogues.

\begin{definition}
We introduce the following two classes.
\begin{enumerate}
  \item
    $\FPE$: problems solvable in $f(k) \cdot 2^{\Ordo(n)} \cdot |x|^{\Ordo(1)}$ time, and
  \item
    $\x{\E}$: problems solvable in $f(k)^n \cdot |x|^{\Ordo(1)}$ time,
    \end{enumerate}
    where $f \colon \mathbb{N} \rightarrow \mathbb{N}$ is a computable function.
\end{definition}

Note that these running times are not symmetric with respect to the two parameters since we expect the parameter $k$ to stay reasonably small, while $n$, naturally, grows rapidly with the size of the instance. Hence, the common trick in parameterized complexity to collapse two parameters into a single parameter $n + k$ is not suitable for defining these classes.

\begin{example}
     Let us first consider finite-domain CSPs. Clearly, if the coarse-grained parameter is the size of the domain, and the fine-grained parameter is the number of variables, then the resulting problem is solvable in $k^n \cdot |x|^{\Ordo(1)}$ time and is thus in $\x{\E}$.  However, this problem is
    {\em not} in $\FPE$ unless the ETH fails (see Theorem~\ref{thm:csp_examples}). For a less trivial $\x{\E}$ example, consider the problem of finding a homomorphism between two graphs $G$ and $H$ over vertices $V_G$ and $V_H$.  If we let $n$ be $|V_G|$,  and  $k$ be the {\em tree-width} of $H$, then the problem is solvable in $(k+3)^n \cdot |x|^{\Ordo(1)}$ time, and if $n = \max \{V_G, V_H\}$ and $k$ is the {\em clique-width} of $H$ then it is solvable in \mbox{$(2k +1)^{n} \cdot |x|^{\Ordo(1)}$} time~\cite{DBLP:journals/mst/FominHK07,DBLP:journals/mst/Wahlstrom11}. Hence, both these parameterisations yield problems in $\x{\E}$.
    Faster algorithms are known for certain types of graphs $H$: if $H$ is the complete clique on $k$ vertices then finding a homomorphism to $H$ is tantamount to finding a $k$-coloring of $G$, which is well-known to be solvable in $\Ordo(2^{|V_G|} \cdot |x|^{\Ordo(1)})$  time~\cite{DBLP:journals/siamcomp/BjorklundHK09} (see Example~\ref{ex:e}). Hence, $k$-coloring parameterized by the number of vertices in the input graph is included in $\FPE$, but is
    arguably not the most interesting example since it is in $\E$ even if $k$ is part of the input.
    Interestingly, if we let the fine-grained parameter $n$ be the tree-width of $G$ and the coarse-grained parameter $k$ be the number of vertices in $H$, then the graph homomorphism problem is solvable in $k^{n+1} \cdot |x|^{\Ordo(1)}$ time via a straightforward dynamic programming algorithm (assuming we are also given a tree decomposition). However, this problem is not solvable in $\Ordo(((k - \varepsilon)^{n+1}) \cdot |x|^{\Ordo(1)})$ for {\em any} $\varepsilon > 0$ under the SETH~\cite{DBLP:journals/siamcomp/OkrasaR21}, and in particular it is not included in $\FPE$. 
\end{example}

\begin{example}
From Example~\ref{ex:e}, equality CSPs are in $\x{\E}$  when parameterized by maximum arity while temporal CSPs are not in $\x{\E}$ with the same parameterization.
More generally, if $\mathcal{B}$ is a so-called {\em partition scheme}, then the CSP problem where constraints are formed by disjunctions of arity at most $k$, using relations from $\mathcal{B}$, is in $\x{\E}$ when parameterized by the maximum variable degree~\cite{jonsson2021}. A non-trivial, related example can also be found in the \textsc{channel assignment problem} which is in $\x{\E}$ when parameterized by the edge-length $k$ of the input graph~\cite{MCDIARMID2003387}. 
\end{example}

Curiously, while $\x{\E}$-algorithms appear to be abundant in the literature, $\FPE$-algorithms are  much rarer and we are not aware of any non-trivial examples, i.e., problems included in $\FPE$ but not in any smaller class such as $\E$. 
However, before we turn to studying concrete problems in these classes we begin by establishing several fundamental properties.
In Section~\ref{sec:parae} we derive $\FPE$ as a class of the form $\para{\E}$ using standard constructions in parameterised complexity and prove that  $\FPE$ can equivalently be defined via running times of the form $n^{f(k)} \cdot 2^{\Ordo(n)}$. In Section~\ref{sec:xe} we investigate the relationship between $\x{\E}$ and $\E$ in greater detail and prove that the former can be defined via piecewise slices of the latter. Last, in Section~\ref{sec:reductions} we consider reductions compatible with bi-parameterized problems, with a particular focus on obtaining relevant reductions for problems in $\x{\E}$.

\subsection{$\FPE$: $f(k) \cdot 2^{\Ordo(n)}$ time}
\label{sec:parae}

In this section we introduce the more desirable class of problems solvable in $f(k) \cdot 2^{\Ordo(n)} \cdot |x|^{\Ordo(1)}$ time, and will see that the resulting class can be obtained via natural constructions in parameterized complexity.

\begin{definition}
  We define the following concepts/classes.
  \begin{enumerate}
    \item
A bi-parameterized problem $P \subseteq \Sigma^* \times \mathbb{N} \times \mathbb{N}$ is in $\E$ after a {\em pre-computation of the parameter} $k$ if there exists an alphabet $\Delta$ and a computable function $f \colon \mathbb{N} \rightarrow \Delta^*$ and a problem $Q \subseteq \Sigma^* \times \mathbb{N} \times \Delta^*, Q \in \E$ where

\[(x,n,k) \in P \Leftrightarrow (x, n, f(k)) \in Q.\]
\item
  $\para{\E}$ is the class of bi-parameterized problems which are in $\E$ after a precomputation of the parameter.
\end{enumerate}    
\end{definition}

As we will now prove, $\para{\E}$ is the natural exponential counter part to $\FPT$, and due to this correspondence we typically write $\FPE$ rather than $\para{\E}$.

\begin{lemma}  
  A bi-parameterized problem $P \subseteq \Sigma^* \times \mathbb{N} \times \mathbb{N}$ with parameters $n$ and $k$ is in $\para{\E}$ ($\FPE$) after a pre-computation of the parameter $k$ if and only if it is solvable in $f(k) \cdot 2^{\Ordo(n)} \cdot |x|^{\Ordo(1)}$ time for computable $f \colon \mathbb{N} \rightarrow \mathbb{N}$.
\end{lemma}

\begin{proof}
The proof is similar to the $\para{C}$ construction in Flum \& Grohe when $C$ is a classical complexity class~\cite{DBLP:series/txtcs/FlumG06}, but we provide a complete proof to ensure that we get the necessary time bound.
  Assume that $P$ is solvable in $f(k) \cdot 2^{cn + \Ordo(1)} \cdot |x|^{\Ordo(1)}$ time for computable $f \colon \mathbb{N} \rightarrow \mathbb{N}$ and $c \geq 1$, and let $A$ denote this algorithm. Thus, the algorithm $A$ takes an instance $x \in \Sigma^*$ together with two parameters $n,k \in \mathbb{N}$ as input. We will prove that $P$ is in $\para{\E}$ after a pre-computation of the parameter $k$. Let $\Delta = \{1, :\}$ where $:$ is a symbol not occurring in $\Sigma$. Define the function $\pi(k) = k : f(k)$ for $k \in \mathbb{N}$, where $k$ and $f(k)$ are written in unary. This function is clearly computable since $f$ is computable. Consider the language $Q \subseteq \Sigma^* \times \mathbb{N} \times \Delta^*$ accepted by the following algorithm $B$:
  \begin{enumerate}
  \item    
    Let $(x, n, y) \in \Sigma^* \times \mathbb{N} \times \Delta^*$ be an instance.
  \item
    Check if the string $y$ is of the form $k : u$ for some ${u \in \{1\}^*}$.
  \item
    If not, answer no, otherwise, continue.
\item  Simulate the algorithm $A$ with $(x, n, k)$ as input for at most $|u| \cdot 2^{cn + \Ordo(1)} \cdot |x|^{\Ordo(1)}$ steps. If $A$ answers yes within this time, answer yes, otherwise answer no. 
\end{enumerate}
We begin by analysing the running time of the algorithm $B$. Since $u$ is part of the input, the simulation of the algorithm $A$ in step (4) can be accomplished with a polynomial overhead, and step (2) takes only linear time with respect to $y$. For correctness, assume that $(x,n,k)$ is a yes-instance of $P$. We claim that $(x,n,\pi(k))$ is a yes-instance of $Q$. In this case the algorithm will reach step (4) and begin the simulation of $A$ for $|u| \cdot 2^{cn + \Ordo(1)} \cdot |x|^{\Ordo(1)}$ steps. 
Assume that $(x,n,\pi(k)) \in Q$. It follows that $A$ answers yes after at most $f(k) \cdot 2^{cn+\Ordo(1)} \cdot |x|^{\Ordo(1)}$ steps, and $(x,n,k)$ must therefore be a yes-instance of $P$.
  
For the other direction, assume that $P$ is in $\para{\E}$ after a pre-computation of the parameter $k$, and let ${Q \in \E \subseteq \Sigma^* \times \mathbb{N} \times \Delta^*}$ and the function $\pi \colon \mathbb{N} \rightarrow \Delta^*$ witness this. Furthermore, let $A$ be an algorithm running in time $2^{\Ordo(n)} \cdot (|x| + |y|)^{\Ordo(1)}$ for input $(x, n, y) \in \Sigma^* \times \mathbb{N} \times \Delta^*$. Now, consider the following algorithm $B$ for $P$:

\begin{enumerate}
\item
  Let $(x, n, k) \in \Sigma^* \times \mathbb{N} \times \mathbb{N}$ be an instance of $P$.
\item
  Call the algorithm $A$ with input $(x, n, \pi(k))$.
\item
  Answer yes if $A$ answers yes, and otherwise no.
\end{enumerate}
This takes $\pi(k) + 2^{\Ordo(n)} \cdot (|x| + |\pi(k)|)^{\Ordo(1)} = \pi(k) + 2^{\Ordo(n)} \cdot (|x|)^{\Ordo(1)}\cdot(|\pi(k)|)^{\Ordo(1)}$ time which can be expressed as $f(k) \cdot 2^{\Ordo(n)} \cdot (|x|)^{\Ordo(1)}$ for some computable $f \colon \mathbb{N} \rightarrow \mathbb{N}$ (which depends on $\pi$).
\end{proof}

We remark that running times of the form ${f(k) \cdot 2^{\Ordo(n)} \cdot |x|^{\Ordo(1)}}$ easily allows us to hide factors of the form $n^{g(k)}$ for computable $g \colon \mathbb{N} \rightarrow \mathbb{N}$. Consider a running time of the form ${n^{f(k)} \cdot 2^{cn + \Ordo(1)} \cdot |x|^{\Ordo(1)}}$ for a computable function $f \colon \mathbb{N} \rightarrow \mathbb{N}$. We claim that it is always bounded by $g(k) \cdot 2^{cn + \Ordo(1)} \cdot 2^n \cdot |x|^{\Ordo(1)}$ where $g(k) = f(k)^{{f(k)}^{f(k)}}$. If $2^n \leq n^{f(k)}$ then $f(k) \geq \frac{n}{\log n}$ and $f(k)^{f(k)} \geq n$, and it follows that $g(k) \geq n^{f(k)}$.

\begin{theorem} \label{thm:parabound}
      A bi-parameterized problem is in $\para{\E}$ if and only if it is solvable in $n^{f(k)} \cdot 2^{\Ordo(n)} \cdot |x|^{\Ordo(1)}$ time for computable $f \colon \mathbb{N} \rightarrow \mathbb{N}$.
\end{theorem}

Note that if we had instead attempted to define $\FPE$ as the single-parameterized problems in $\E$ with respect to a parameter $n + k$ in $\E$ then the factor $n^{f(k)}$, necessary for the CSP$(\mathcal{A})$ algorithm in the main paper, would be prohibited. A similar attempt of defining $\FPE$ as the single-parameterized problems solvable in $f(k) \cdot |x|^k \cdot 2^{\Ordo(|x|)}$ time would also be unsuitable since $|x|$ is not a meaningful parameter for single-exponential time algorithms.

\subsection{$\x{\E}$: $f(k)^{n}$ time}
\label{sec:xe}
In this section we investigate bi-parameterized problems solvable in $f(k)^{n} \cdot |x|^{\Ordo(1)}$ time for computable $f \colon \mathbb{N} \rightarrow \mathbb{N}$.

\begin{definition}
  Let $P \subseteq \Sigma^* \times \mathbb{N}$ be a bi-parameterized problem.
    For $k \geq 1$ the {\em $k$th slice} of $P$ is the parameterized problem
    $P_k = \{(x,n) \mid (x, n, k) \in P\} \subseteq \Sigma^{*} \times \mathbb{N}$.
\end{definition}

\begin{definition} \label{def:xe} We define the following complexity classes.
  \begin{enumerate}
    \item
      $\x{\E}_{\mathrm{NU}}$ is the set of all bi-parameterized problems $P$ such that $P_k \in \E$ for each $k \geq 1$.
    \item
      $\x{\E}_{\mathrm{U}}$ is the set all bi-parameterized problems where (1) each slice $P_k \in \E$ can be solved by an algorithm $A_k$, and (2) there exists a computable function ${g \colon \mathbb{N} \rightarrow \{(A_1, c_1), (A_2, c_2) \ldots\}}$ which associates each $k$ with the algorithm $A_k$ for $P_k$ running in $2^{c_kn + \Ordo(1)}$ time. 
    \end{enumerate}
  \end{definition}

  Since we in this paper do not care about undecidable problems we henceforth only consider the uniform class $\x{\E}_{\mathrm{U}}$ and write $\x{\E}$ instead of $\x{\E}_{\mathrm{U}}$.
  
    \begin{lemma}  
      A bi-parameterized problem $P$ is in $\x{\E}$ if and only if it is solvable in $f(k)^n \cdot |x|^{\Ordo(1)}$ time for some computable function $f \colon \mathbb{N} \rightarrow \mathbb{N}$.
    \end{lemma}

    \begin{proof}
      For the first direction, let $(x,n,k)$ be an instance of the problem. We begin by determining the algorithm $A_k$ and the function $f_k$ which takes $g(k)$ time for a fixed, computable function $g$. Next, we run the algorithm $A_k$ with the input $(x,n)$ for $2^{c_k n + \Ordo(1)}$ steps, and answer yes if and only if it accepts. Hence, we obtain a running time $g(k) + 2^{c_k n + \Ordo(1)} \cdot |x|^{\Ordo(1)}$  time, which can be expressed as $f(k)^{n} \cdot |x|^{\Ordo(1)}$ for a computable $f$ which depends on $g$ and $c_k$.
      For the other direction we simply observe that a running time of the form $f(k)^n \cdot |x|^{\Ordo(1)}$ clearly implies that each slice $P_k \in \E$ since $k$ is fixed.
    \end{proof}

\subsection{Reductions for Bi-Parameterized Problems}
\label{sec:reductions}    
    We proceed by describing reductions compatible with problems in $\FPE$. The complicating factor is that we have two parameters to take into consideration: the parameter $k$ which may increase substantially as long as it is bounded by a computable function in $k$, and the parameter $n$ which is only allowed to increase linearly --- in order to preserve running times of the form $2^{\Ordo(n)}$.

\begin{definition}
  Let $P \subseteq \Sigma^* \times \mathbb{N} \times \mathbb{N}$ and $Q \subseteq \Delta^* \times \mathbb{N} \times \mathbb{N}$ be two bi-parameterized problem. An {\em $\FPE$-reduction} with respect to $n$ is  is a function $R \colon \Sigma^* \times \mathbb{N} \times \mathbb{N} \rightarrow \Delta^* \times \mathbb{N} \times \mathbb{N}$ such that

  \begin{enumerate}
  \item
    $(x, n, k) \in P$ if and only if $R((x, n, k)) \in Q$.
  \item
    There exists $c \in \mathbb{N}$ such that $n' = c \cdot n + \Ordo(1)$ for all instances $(x, n, k)$ and $R((x,n,k)) = (x',n',k')$ of $P$ and $Q$. 
  \item
    There exists a computable function $f \colon \mathbb{N} \rightarrow \mathbb{N}$ such that $R((x, n, k))$ is computable in $f(k) \cdot 2^{\Ordo(n)} \cdot |x|^{\Ordo(1)}$ time for any instance $(x,n,k)$ of $P$.
    \item
      There exists a computable function $g \colon \mathbb{N} \rightarrow \mathbb{N}$ such that $k \leq g(k')$ for all instances $(x, n, k)$ and ${R((x, n, k)) = (x', n', k')}$ of $P$ and $Q$.
  \end{enumerate}    
\end{definition}

 We write $P \parared Q$ if $P$ is $\FPE$-reducible to $Q$. If $P$ and $Q$ are two single-parameterized problems then we say that $P$ is {\em $\E$-reducible} to $Q$, written $P \ered Q$, if $P \times \{0\} \parared Q \times \{0\}$. Thus, an $\E$-reduction is a single-exponential time reduction which may increase the number of variables by a constant factor.

\begin{theorem} \label{thm:reductions} 
    The following statements hold.
    \begin{enumerate}
        \item 
        $\E$ is closed under $\E$-reductions.
        \item
        $\FPE$ is closed under $\FPE$-reductions.
    \end{enumerate}
\end{theorem}

\begin{proof} 
Since the first statement can be seen as a special case of the second statement we only prove the former. Hence, assume that $P \in \FPE$ and that $Q$ is a bi-parameterized problem such that $Q \parared P$. We will prove that $Q \in \FPE$. Since ${P \in \FPE}$ there exists computable $f_P \colon \mathbb{N} \rightarrow \mathbb{N}$ and constants $c_P$ and $D_P$ such that $P$ is solvable in $f_P(k_P) \cdot 2^{c_Pn_P + \Ordo(1)} \cdot |x|^{d_P}$ time for any instance $(x, n_P, k_P)$ of $P$. Let $R$ be the $\FPE$-reduction from $Q$ to $P$, with respect to (1) a constant $c \geq 1$ which linearly bounds the parameter $n$ and (2) a computable function $f_R \colon \mathbb{N} \rightarrow \mathbb{N}$ and  constants $c_R$ and $d_R$ such that $R((x,n_Q,k_Q))$ can be computed in $f_R(k_Q) \cdot 2^{c_R n_Q + \Ordo(1)} \cdot |x|^{d_R}$ time for any instance $(x,n_Q,k_Q)$, and (3) a computable function $g_R \colon \mathbb{N} \rightarrow \mathbb{N}$ which bounds the parameter $k$, i.e., $k_Q \leq g_R(k_Q)$ for any instance $(x,n_Q,k_Q)$. Let $(x, n_{Q}, k_{Q})$ be an instance of $Q$. Computing $R((x, n_Q, k_Q))$ then takes $f_R(k_Q) \cdot 2^{c_R \cdot n_Q + \Ordo(1)} \cdot |x|^{d_R}$ time, and results in an instance $(x', n_P, k_P)$ of $P$ where $n_P = c \cdot n_Q + \Ordo(1)$, $k_P = g_R(k_Q)$, and where ${|x'| \leq f_R(k_Q) \cdot 2^{c_R \cdot n_Q + \Ordo(1)} \cdot |x|^{d}}$. Hence, the instance is solvable in $f_P(g_R(k_Q)) \cdot 2^{c \cdot c_P \cdot n_Q + \Ordo(1)} \cdot (f_R(k_Q) \cdot 2^{c_R \cdot n_Q + \Ordo(1)} \cdot |x|^{d_R})^{d_P + 1}$ time, which yields an $\FPE$ running time bounded by
\[f_Q(k) \cdot 2^{c_Q \cdot n + \Ordo(1)} \cdot |x|^{d_R(d_P + 1)}\]
for $f_Q(k) = f_P(g_R(k_Q)) \cdot f(k_Q)^{d_P + 1}$ and a constant ${c_Q \leq c \cdot c_P + c_R \cdot (d_P + 1)}$, and $Q$ is thus in $\FPE$.
\end{proof}

We refrain from formally introducing $\x{\E}$-reductions since the following theorem, which follows immediately from Definition~\ref{def:xe} and Theorem~\ref{thm:reductions}, already provides a straightforward method to rule out inclusion in $\x{\E}$.

\begin{theorem} \label{thm:non_inclusions}
    Let $Q$ be a bi-parameterized problem. 
    \begin{enumerate}
        \item 
        If there exists a slice $Q_{k_0} \notin \E$ then $Q \notin \x{\E}$.
        \item
        If there exists a single-parameterized problem $P \notin \E$ and a slice $Q_{k_0}$ such that $P \ered Q_{k_0}$ then $Q \notin \x{\E}$.
    \end{enumerate}
\end{theorem}

\section{Parameterized Problems in $\FPE$ and $\x{\E}$} \label{sec:algorithms}
We now analyse $\FPE$ and $\x{\E}$ in greater detail and begin in Section~\ref{sec:fin} by considering finite-domain CSPs with parameters such as {\em domain size} and {\em degree}. 
In Section~\ref{sec:partial} we consider the {\em effective width} of a partial order as a coarse-grained parameter for \textsc{partially ordered time}, and construct a non-trivial $\x{\E}$-algorithm. Last, in Section~\ref{sec:allen}, we use the maximum number of possible interval {\em overlaps} as a parameter for \textsc{Allen's interval algebra}, resulting in the first non-trivial example of a problem in $\FPE$.

\subsection{Finite-Domain CSPs} 
\label{sec:fin}
For a finite-domain CSP instance $(V,C)$ we consider the following parameters ($\ar(R)$ is the arity of a relation $R$).

\begin{enumerate}
  \item
    $\text{dom}((V,C)) = \bigcup_{(d_1, \ldots, d_{\ar(R)}) \in R, R(\mathbf{x}) \in C} \{d_1, \ldots, d_{\ar(R)}\}$,
  \item
    $\text{max-arity}((V,C)) = \text{max} \{\ar(R) \mid R(\mathbf{x}) \in C\}$,
  \item
    $\text{max-degree}((V,C)) = \text{max} \{\text{degree}(x,C) \mid x \in V\}$ where $\text{degree}(x, C) = |\{R(\mathbf{x}) \in C \mid x \text{ occurs in } \mathbf{x}\}|$,
  \item
    $\text{max-cardinality}((V,C)) = \text{max} \{|R| \mid R(\mathbf{x}) \in C\}$,    
\end{enumerate}

For each such parameter $f$ we write $p$-$f$-CSP for the bi-parametrized problem obtained by letting the coarse-grained parameter $k$ equal $f((V,C))$ and the fine-grained parameter $n$ equal the number of variables $|V|$ (for a CSP instance $(V,C)$). For example, $p$-dom-CSP denotes the CSP problem where the coarse-grained parameter $k$ is the domain
and the fine-grained parameter $n$ is the number of variables. Before proving the main result we introduce the following problem from Traxler~\cite{traxler2008}. 

\multiproblemDef{$p$-SparseBinCSP}
{A natural number $d$, a set of variables $V$, and a set of constraints $C$ where $\text{max-arity}((V,C)) = 2$ and $\text{max-degree}((V,C)) \leq 3 \cdot d^2$ where $d = \text{dom}((V,C))$.}
{$n = |V|$.}
{Does there exist a function $V \rightarrow D$ which satisfies all constraints, i.e., $(f(x_1), \ldots, f(x_m)) \in R$ for every $R(x_1, \ldots, x_m) \in C$?}

We let $p$-dom-SparseBinCSP be the bi-parameterized variant of $p$-SparseBinCSP where the parameter $k$ bounds the domain. This leads to our first problem which is provably not in $\FPE$ (unless the ETH fails).

\begin{theorem} \label{thm:traxler}  
Assume that the ETH is true. Then:
\begin{enumerate}
    \item 
    $p$-SparseBinCSP is not in $\E$, and
    \item
    $p$-dom-SparseBinCSP is not in $\FPE$.
\end{enumerate}
\end{theorem}

\begin{proof}
 We begin by proving the following claim: for every $c \geq 1$ there exists a $d \geq 1$ such that $p$-dom-SparseBinCSP$_d$ is not solvable in $2^{cn + \Ordo(1)} \cdot |x|^{\Ordo(1)}$ time (recall that $P_d$ is the $d$th slice of $P$ with respect to $d$) . First, from Traxler~\cite{traxler2008} it follows that any algorithm which solves $p$-dom-SparseBinCSP$_d$ in $2^{cn}$ time must satisfy $\log d \cdot c' \leq c$ for a universal constant $c' \geq 1$\footnote{Throughout the paper Traxler assumes the slightly stronger statement {\em randomized ETH} (r-ETH) but only the ETH is needed for the binary, sparse CSP problem.}.
    Furthermore, observe that there for every $c$ exists $c' \geq c$ such that $2^{cn + \Ordo(1)} \cdot |x|^{\Ordo(1)}$ is bounded by $2^{c'n}$, i.e., that a running time of the form  $2^{cn}$ easily allows us to hide polynomial factors of the input size, and additive constants in the exponent, since the number of constraints is linearly bounded by the number of variables in the $p$-SparseBinCSP$_d$ problem, and since $d$ is fixed. Thus, the claim easily follows. 
    
    Next, assume that $p$-SparseBinCSP is in $\E$, i.e., that there exists $c \geq 1$ such that the problem is solvable in $2^{cn + \Ordo(1)} \cdot |x|^{\Ordo(1)}$ time. It follows that $p$-dom-SparseBinCSP$_d$ for every fixed $d$ is solvable in $2^{cn + \Ordo(1)} \cdot |x|^{\Ordo(1)}$ time, which contradicts the above claim.
    
    Last, assume that $p$-dom-SparseBinCSP is in $\FPE$, i.e., that there exists a computable function $f \colon \mathbb{N} \rightarrow \mathbb{N}$ and a $c \geq 1$ such that the problem is solvable in $f(k) \cdot 2^{cn} \cdot |x|^{\Ordo(1)}$ time. Thus, pick a $d \geq 1$ such that $p$-dom-SparseBinCSP$_d$ is not solvable in $2^{cn} \cdot |x|^{\Ordo(1)}$ time (via the above claim). It follows that the proposed algorithm solves $p$-dom-SparseBinCSP$_d$ in $f(d) \cdot 2^{cn} \cdot |x|^{\Ordo(1)}$ time, and, since $d$ is fixed, the expression $f(d)$ can easily be hidden in the polynomial factor $|x|^{\Ordo(1)}$.
\end{proof}

While not difficult to prove, Theorem~\ref{thm:traxler} provides a good starting point for analysing the remaining parameterisations of finite-domain CSPs.

\begin{theorem} \label{thm:csp_examples} 
Assume that the ETH is true. Then:
  \begin{enumerate}
  \item
    $p$-dom-CSP is in $\x{\E}$ but not in $\FPE$,
  \item
    $p$-max-arity-CSP is not in  $\x{\E}$,
  \item
    $p$-max-degree-CSP is not in $\FPE$, and
  \item
    $p$-max-cardinality-CSP is in $\x{\E}$ but not in $\FPE$.
  \end{enumerate}
\end{theorem}

\begin{proof}
  We prove each statement in turn.
\begin{enumerate}
    \item The inclusion in $\x{\E}$ is trivial by enumerating all possible assignments with respect to the parameter $k$. For the lower bound, assume that the problem is solvable in $f(k) \cdot 2^{cn} \cdot |x|^{\Ordo(1)}$ time for some constant $c$. We then have a trivial $\FPE$-reduction from $p$-dom-SparseBinCSP to $p$-dom-CSP, which (using Theorem~\ref{thm:reductions}) contradicts Theorem~\ref{thm:traxler}.
    \item
    It is easy to see that $p$-SparseBinCSP is a restricted case of the slice $p$-max-arity-CSP$_2$. Hence, we have a trivial $\E$-reduction from the former to the latter ($p$-SparseBinCSP $\ered$ $p$-max-arity-CSP$_2$) which by Theorem~\ref{thm:non_inclusions} and Theorem~\ref{thm:traxler} implies non-inclusion in $\x{\E}$ under the ETH.
    \item
     We have a simple reduction from $p$-dom-SparseBinCSP to $p$-max-degree-CSP which changes the parameter $d$ to $3d^2$, i.e., $p$-dom-SparseBinCSP $\parared$ $p$-max-degree-CSP.
    \item
    To show inclusion in $\x{\E}$ we use a branching algorithm: pick a variable $x_i$ occurring in a constraint $R(x_1, \ldots, x_i, \ldots, x_n)$. Since $|R| \leq k$ we have to try at most $k$ possibilities for $x_i$ and obtain a branching factor of $k$. Since the depth of the tree is at most $n$ we obtain a running time of $k^n \cdot |x|^{\Ordo(1)}$. For non-inclusion in $\FPE$, note that a binary relation over a domain with $d$ elements contains at most $d^2$ tuples. It then follows that we have an $\FPE$-reduction from $p$-dom-SparseBinCSP to $p$-max-cardinality-CSP which changes the parameter $d$ to $d^2$, i.e., $p$-dom-SparseBinCSP $\parared$ $p$-max-cardinality-CSP.
\end{enumerate}
\end{proof}

It may be interesting to note that all these bound are tight, with the possible exception of $p$-max-degree-CSP, where we currently do not know whether it is included in $\x{\E}$.

\subsection{Partially Ordered Time}
\label{sec:partial}

The first major temporal reasoning problem we consider is the \textsc{partially ordered time} problem where we 
use a form of \textit{width} as the coarse-grained parameter.

\begin{definition}
A partial order $P = (S, \leq_P)$ is said to have \emph{\ew} $k$ if $S$ can be partitioned into at most $k+2$ disjoint subsets $(S_1,\dots,S_k,S_<,S_>)$ such that:
\begin{enumerate}
    \item
        at least two of the sets $(S_1,\dots,S_k,S_<,S_>)$ are non-empty, unless all elements are equal,
    \item
        for all $x\in S_<$ there is a $S_i$ such that for all $y\in S_i$ then $x\leq_Py$ and not $y\leq_Px$,
    \item
        for all $x\in S_>$ there is a $S_i$ such that for all $y\in S_i$ then $x\geq_Py$ and not $y\geq_Px$,
    \item
        for all pairs $x\in S_<$, $y\in S_>$ with $x\leq_Py$, then there is a $S_i$ such that for all $z\in S_i$ then $x\leq_P z \leq_Py$, 
    \item
        all partial orders $P[S_1],\dots,P[S_k],P[S_<],P[S_>]$ have {\ew} $k$, and
    \item
        for every pair $S_i,S_j\in\{S_1,\dots,S_k\}$, $i\neq j$, all $x\in S_i$ and all $y\in S_j$ are incomparable in $P$.
\end{enumerate}
The collection of sets $S_1,\dots,S_k$ is called the \emph{waist} of $P$. 

\end{definition}
 {\Ew} is a weaker constraint than width (the size of the largest anti-chain), but is likely a stronger condition  than the \emph{dimension} of a partial order, i.e., the least number of total orders whose intersection equals the partial order.

\multiproblemDef{Partially Ordered Time of {\ew} $k$}
{A partially ordered time instance $(V,C)$.}
{$n = |V|$, $k \in \mathbb{N}$.}
{Is there a partial order $P$ with {\ew} $k$ which satisfies the instance?}

We introduce the following two definitions to simplify the notation in the forthcoming algorithm. 

\begin{definition}
Let \emph{$\mathbf{B}(n,k)$} be the value of the $k$-th bit of the bit-string representing the integer $n$.
\end{definition}

\begin{definition}
Given a binary relation $\odot$, a variable $x$ and a set $V$, we say that $x\odot V$ if $x\odot v$ for all $v\in V$.
\end{definition}

This notation easily extends to pairs of sets.
The basic idea is then to recursively construct a partial order $P$ of {\ew} $k$ by enumerating all possible waists and relations between the other variables and said waist.
This enumeration is then combined with dynamic programming to keep track of already solved subproblems.
First, observe that each element not in a waist can (independently) be either $<$ (or $>$) or incomparable to each of the $k$ sets of said waist. If we are relating to two waists at once we then obtain $4^k$ possibilities.

\begin{theorem}
Partially Ordered Time of {\ew} $k$ is solvable in $16^{kn} \cdot |x|^{\Ordo(1)}$ time, and is hence in $\x{\E}$.
\end{theorem}

\begin{proof}
For a \textsc{Partially Ordered Time} instance $(V,C)$ of {\ew} $k$
we define a recursive function $\mathcal{X}$, taking as input $4^k$ pairwise disjoint subsets of $V$, $S_1,\dots,S_{4^k}$ as follows.
If only one $S_i \neq \emptyset$, check if the assumption that all variables of that set are equal is consistent with $C$. This can easily be done in polynomial time.
If yes, the instance is accepted, otherwise we continue as follows.
We introduce two temporary waists of size $k$: $\{T^<_1,\dots, T^<_k\}$ and $\{T^>_1,\dots,T^>_k\}$.
Assume that \mbox{$T^<_i <_P S_j \Leftrightarrow \mathbf{B}(j,2i-1)=1$} and \mbox{$T^>_i<_P S_j \Leftrightarrow \mathbf{B}(j,2i)=1$}. 
Partition $\bigcup S_i$ into the new sets $N^W_i$ for $i\in\{1,\ldots,k\}$, and $N^<_i$, $N^>_i$ for \mbox{$i\in\{1,\ldots,4^k\}$}, 
such that no constraint is broken under the following assumptions:

\begin{enumerate}
\item if $\mathbf{B}(i,2j)=1$ then $N^<_i<_PN^W_j$ and $N^>_i <_P T^>_j$, 
\item if $\mathbf{B}(i,2j-1)=1$ then $T^<_j<_P N^<_i$ and $N^W_j<_PN^>_i$,
\item all pairs $N^W_i$ and $N^W_j$ are incomparable if $i\neq j$.
\end{enumerate}
    
Note that this means that for each $i$ there is a $j$ such that $N^W_i\subseteq S_j$, and since every variable not in a waist must relate to the waist, some sets must always be empty.
For each partition, recursively check that $I[N^W_i]$, $I[\bigcup N^<_i]$ and $I[\bigcup N^>_i]$ are 'yes'-instances for all $i$ (when they exist) under the assumptions just made. 
If one partition is found that is consistent with $C$ and $\mathcal{X}(N^<_1,\dots,N^<_{4^k})$, $\mathcal{X}(N^W_1,\emptyset,\dots)$, $\dots$, $\mathcal{X}(N^W_k,\emptyset,\dots)$ and $\mathcal{X}(N^>_1,\dots,N^>_{4^k})$ all accept, then accept $\mathcal{X}(S_1,\dots,S_{4^k})$.

For $|\bigcup S_i|\leq1$ the algorithm is clearly correct.
Assume that the algorithm is correct for all $|\bigcup S_i|<n$.
Now, for correctness of $|\bigcup S_i|=n$, we start with soundness.
We brute-force enumerate all partitions and reject a partition if there is any constraint not consistent with said partition.
For every constraint $c(x,y)\in C$, there will at some point be a partition where $x$ and $y$ are not part of the same recursive call to $X$.
At that point we will know the exact relation between $x$ and $y$, and hence we can check if this relation is accepted by the constraint $c(x,y)$.
Hence, if we answer 'yes', all constraints will be satisfied.
To see that we find an partial order, the question boils down to knowing transitivity holds correctly.
Since we reject all partitions that breaks our assumptions, we will never accept $x<_Py$ if we have earlier assumed that $z<_Px$ and that $y$ and $z$ are incomparable.
The same is true for the symmetric case when $z<_Px$ is assumed by some earlier partition and $x<_Py$ by the current partition. 
Also $x\in N^W_i$ and $y\in N^W_j$ would only be accepted it $x$ being incomparable to $y$ is accepted by our constraint $c(x,y)$.
Hence, these assumptions are enough to ensure that we actually obtain a partial order.
For {\ew} we follow the definition and in each step make sure that:
everything not part of the waist is related to some part of the waist,
all parts of the waist are incomparable to every other part of the waist, and 
that all of $I[\bigcup N^<_i]$, $I[N^W_1]$, $\dots$, $I[\bigcup N^W_k]$, and $I[\bigcup N^>_i]$ have {\ew} $k$.
Hence, we will only answer 'yes' if we are given a 'yes'-instance.
For completeness, assume there is some partial order $P_A$ of {\ew} $k$ satisfying $I$.
Since $P_A$ exists, the waist $N^W_1,\dots,N^W_k$ exists and we will find it and the correct partitions to give as inputs to the subproblems, when brute-force testing all possible partitions.
Since the algorithm is assumed to be correct for every problem smaller than $n$, it must be correct (and thus complete) for these subproblems.

By standard dynamic programming practices of storing results for $X$ for previously calculated inputs
the running time of the algorithm is dominated by the number of partitions times the cost for each step.
At each step, the only non-polynomial cost is the partitioning, which can be done in $(k+2\cdot4^k)^n$ different ways, yielding this as a multiplicative factor to the total number of possible inputs, and hence a total run-time of $(4^{k}\cdot(k+2\cdot4^k))^n \cdot |x|^{\Ordo(1)}$, or $16^{kn} \cdot |x|^{\Ordo(1)}$, showing inclusion in $\x{\E}$.
\end{proof}

Let us remark that this algorithm solves the more general problem of counting the number of solutions.
Naturally, it is straightforward to modify the algorithm to actually return a solution,  e.g., via  back-pointers, or any other standard methods to iteratively build a solution.

\subsection{Allen's Interval Algebra}
\label{sec:allen}

The second major qualitative reasoning problem we study is a variant of \textsc{Allen's interval algebra} where we restrict the number of intervals any interval may \emph{overlap} with.

\begin{definition}
We say that two intervals $I$ and $J$ overlap if there is a point $x$ such that $I^-<x<I^+$ and $J^-<x<J^+$.
\end{definition}

\multiproblemDef{$k$-CSP$(\mathcal{A})$}
{A CSP$(\mathcal{A})$ instance $I=(V,C)$.}
{$n = |V|$, $k \in \mathbb{N}$.}
{Is there a satisfying assignment where no interval overlaps with $k$ or more intervals?}

We will prove that $k$-CSP$(\mathcal{A})$ is solvable in roughly $(2kn)^{2k}\cdot 2^n$ time and space, and thus prove membership in $\FPE$. Our algorithm makes use of a relationship between satisfying assignments and so-called {\em ordered partitions}.

\begin{definition}
A finite sequence of non-empty finite sets $(S_1,\ldots,S_\ell)$ is an \emph{ordered partition} of a set $S$ if $S_1,\ldots,S_\ell$ is a partitioning of $S$.
The \emph{ranking function} $r$ is the function $r \colon S \rightarrow \{1, \ldots , \ell\}$ such that $r(x)=i$ for every $x \in S_i$ and every $i\in \{1, \ldots , \ell\}$.
The number of unique ordered partitions for a set of size $n$ is given by the \emph{Ordered Bell Number} $\mathrm{OBN}(n)$ and is strictly less than the number of possible ways to arrange $n$ items into $n$ different sets if $n>1$, $n^n$.
\end{definition}

It is known that CSP$(\mathcal{A})$ over $V$ is satisfiable if and only if there exists an ordered partition of $V^- \cup V^+$ and a ranking function $r$, such that, for any relation \mbox{$\odot \in \{<,>,=\}$}, \mbox{$r(x) \odot r(y)$} if and only if $f(x) \odot f(y)$ for all \mbox{$x,y\in V^- \cup V^+$}~\cite{lagerkvist2021b}. To simplify the notation we allow ordered partitions to (temporally) contain empty sets $S_i$ and we then assume that $(S_1,\ldots,S_i,\ldots,S_\ell)=(S_1,\ldots,S_{i-1},S_{i+1},\ldots,S_\ell)$.
We will explicitly make sure this step is taken, but allowing this slight bending of the definition lets us avoid some unnecessarily complex steps.

\begin{definition}
Given a sequence (such as e.g. an ordered partition) $(x_1,\dots,x_i)$ let the notation $(x_1:x_i)$ denote this sequence.
Let $(x_1:x_i-x_j)=(x_1,\dots,x_{j-1},x_{j+1},\dots,x_i)$, $j\in\{1,\ldots,i\}$.
\end{definition}

\begin{theorem}
$k$-CSP$(\mathcal{A})$ is solvable in $(2nk)^{2k}\cdot 2^n \cdot |x|^{\Ordo(1)}$ time, and is hence in $\FPE$.
\end{theorem}
\begin{proof}
The algorithm is reminiscent of the one presented in~\cite{lagerkvist2021b} and works by the following principle:
by introducing intervals from smallest starting-point, to largest end-point, we only need to keep track of intervals already passed, current open intervals, how many other intervals these has already overlapped with, and intervals not yet opened. 
Given an $k$-CSP($\mathcal{A}$) instance $I=(V,C)$, for inputs $V^-_1$ (a set), $X=(X_1:X_i)$ (a partial order), $v=(v_1:v_i)$ (a collection of values) and $V^-_2$ (a set) we then define a recurrence relation $R(V^-_1,X,v,V^-_2)$ as follows:
\begin{itemize}
    \item $R(V^-_1,(X_1:X_i-X_j),(v_1:v_i-v_j),V^-_2)$, if $X_j=\emptyset$, 
    \item $0$, if $v_j<0$ for any $j\in\{1,\dots,i\}$,
    \item $1$, if $V^-_1 = V^-$, and 
    \item $\sum_{x^-\subseteq\bigcup X_j, y \subseteq V^-_2 } f(V^-_1,X,v,V^-_2,x^-,y)$ otherwise.
\end{itemize}
If two cases are applicable for the same input, the first, as ordered here above, have the highest priority.

Using a set $x^-$ and the sequence \mbox{$X=(X_1 : X_i)$} define the sets

\[x^+=\{x^+_j\,|\,(x^-_j,x^+_j)\in V \land x^-_j\in x^-\}\]

and

\[X^+=\{x^+_j\,|\,(x^-_j,x^+_j)\in V \land x^-_j\in \bigcup X_l\setminus{x^-}\},\]

the sequences

\[X'=(X_1\setminus{x^-}:X_i\setminus{x^-},y)\]

and

\[v'=(v_1-|x^-|:v_i-|x^-|,k-|\bigcup X_j\setminus{x^-}|)\]

and the function $f$ as

\[
f(V^-_1,X,v,V^-_2,x^-,y) = R(V^-_1\cup x^-,X',v',V^-_2\setminus{y})
\]

if and only if the assumptions \mbox{$(X_1:X_i,y\cup x^+)$} and $u<v$ for all $u\in\bigcup X_j\cup y\cup x^+$, $v\in V^-_2\cup X^+$ is consistent with $C$.
I.e. for any constraint $c^+(u,v)=c^-(u,v)\in C$ ($\mathcal{A}$ is closed under symmetry, so the constraints are equal) with $v^-\in \bigcup X_j \cup V^-_2$ and $u^+\in x^+$ then the assumption satisfies $c^+(u,v)$.
Otherwise $f$ outputs $0$.

Our algorithm for the decision problem is simply to check if $R(\emptyset,(),(),V^-)>0$.
To see that this approach is sound we follow a sequence of $R$-recurrences from $R(\emptyset,(),(),V^-)$ to $R(V^-,(),(),\emptyset)$ such that in each step we have that $R>0$.
Take the ordered partitions that are part of this sequence, and merge them into a single ordered partition $(Y_1:Y_n)$ such that if $x<y$ in some $(X_1:X_l, V^-_2)$ for some \mbox{$R(V^-_1,(X_1:X_l),v,V^-_2)$}, then $x<y$ in $(Y_1:Y_n)$.
Since $R$ works by moving starting-points from the set $V^-_2$ to $X$ and then further to $V^-_1$, while keeping track of relevant previous intervals
$(Y_1:Y_n)$ exists and is uniquely defined for each such sequence.
The positions of the end-points in $(Y_1:Y_n)$ is implicit from the relation between $x^-$ and $x^+$ in our definition of $f$.
Let $u^-\in Y_{i^-}$, $u^+\in Y_{i^+}$, $u'^-\in Y_{j^-}$ and $u'^+\in Y_{j^+}$.
Assume ${i^+}\leq {j^+}$, then for each constraint $c(u,u')\in C$ in our sequence of $R$-recurrences, there is an $R(V^-_1,X,v,V^-_2)$ followed by $R(V'^-_1,X',v',V'^-_2)$ such that $u^-\in V'^-_1$ and $u^-\not \in V^-_1$.
From these two inputs, we know the relations between $i^-$, $i^+$, $j^-$ and $j^+$ (since $j^- < j^+$ must hold).
Since we know these relations, and since the sequence exists, $f$ must have accepted that $c(u,u')$ is holds.
This is true for all $c\in C$, so $(Y_1:Y_n)$ must be a solution for $I$.
To verify that this solution also fulfills the overlap requirement, we at each step have $v$ which is handled such that no value in $V$ never drops below $0$.
Hence we have a 'yes'-instance if $R(\emptyset,(),(),V^-)>0$.

For completeness, take an ordered partition \mbox{$Y=(Y_1:Y_i)$} over $V^-\cup V^+$ satisfying our arbitrary instance $I$.
Given this ordered partition, iterative construct a sequence of $R(V^-_1,X,v,V^-_2)$ and and choose $x^-$ and $y$ such that \mbox{$x^+=Y_j\cap V^+$} and $y = Y_j\cap V^-$, starting from $R(\emptyset,(),(),V^-)$ and $Y_1$, working towards $Y_i$ and $R(V^-,(),(), \emptyset)$. 
Now, observe that for \mbox{$X=(X_1:X_{i'})$} we have that $X_1\subseteq Y_l$, $\dots$, $X_j\subseteq Y_{l+i'-1}$, and \mbox{$V^-_2 = V^-\cap \bigcup_{u=l+i'}^i Y_{u}$}, which means that when placing $Y_l$, we will know the relations between all intervals ending in $Y_l$, and all intervals ending later.
Since these relations are the same as those in $Y$, $f$ will always return a value greater than $0$.
Similarly, $v$ will never contain a value less than $0$, since the numbers in $v$ simply counts how many overlaps have occurred for corresponding intervals.
The ordering $(Y_1:Y_n)$ will also be tested as every sub-ordering of size at most $k$ will be found by brute-force enumeration.
So, if given a 'yes'-instance, we will reach $R(V^-,(),(), \emptyset)=1$ and so every other $R$ in our sequence (including $R(\emptyset,(),(),V^-)$) we will have $R(V^-_1,X,v,V^-_2)>0$.
Hence, this approach is correct.

For the complexity there are fewer than $kn^k2^n \cdot \mathrm{OBN}(k)$ ways to partition $V^-$ into two unbounded sets and at most $k$ non-empty sets of a combined size of at most $k$.
Additionally there are at most $k^k$ ways these size bound partitions can overlap with previous intervals, and this also superimposes the factor $\mathrm{OBN}(k)$ since the ordering is precisely defined by the number of previously overlapped intervals.
Finally, there are at most $(2n)^k$ branches in each step, with a polynomial overhead.
This gives us a complexity of $(2nk)^{2k} 2^n \cdot |x|^{\Ordo(1)}$, proving membership in $\FPE$.
\end{proof}

This algorithm, similarly to the one in Section~\ref{sec:partial}, solves the more general problem of counting the number of solutions (i.e., the number of satisfying partial orders or ordered partitions)

\section{Concluding Remarks} \label{sec:discussion}

We explored how a multivariate complexity analysis can generate interesting cases with significant complexity improvements compared to the classical case of single parameters. This led to the introduction of the classes $\FPE$ and $\x{\E}$ which turned out to be the natural exponential-time analogoues of the well-known classes $\FPT$ and $\x{\mathbf{P}}$. We proved several fundamental properties of these classes and gave examples of $\x{\E}$ and $\FPE$ from qualitative reasoning.
These algorithms are significantly faster than existing methods (provided the coarse-grained parameters stay relatively small) and constitute an important breakthrough for single-exponential time algorithms in qualitative reasoning. 

\subsection*{Systematically classifying $\x{\E}$ and $\FPE$.}
The overarching open question is to classify problems and parameters in terms of $\x{\E}$ or $\FPE$ membership (or neither, under assumptions such as the ETH). Here, qualitative spatial reasoning problems (such as the \textsc{region connection calculus}) seem to be a promising continuation since the NP-hard cases in general are not known to be solvable in single-exponential time. More generally, is it possible to find canonical parameters which result in $\FPE$ algorithms, similarly to how parameters such as tree-width almost always results in $\FPT$ algorithms? Identifying such parameters could open up entirely new algorithmic approaches for qualitative reasoning, in particular, and infinite-domain CSPs, in general.

\subsection*{Limited versus unlimited equality.}
The algorithm in Section~\ref{sec:allen} for $k$-CSP($\mathcal{A}$) has similarities to the one in~\cite{MCDIARMID2003387} for the \textsc{channel assignment problem} with bounded edge-length, making the two problems interesting to compare.
One interpretation is that $k$-CSP($\mathcal{A}$) belongs to $\FPE$ because the problem bounds the number of intervals that may be equal (since they overlap), which is not true for \textsc{channel assignment}, which seemingly prevents our approach from proving $\FPE$ membership.
However, if one limits the size of equivalence classes in the same sense in the \textsc{channel assignment problem}, this new problem would naturally fall into $\FPE$, using just a slight modification of the original algorithm. Is this a problem specific behaviour, or is it a deeper difference between $\FPE$ and $\x{\E}$?

\subsection*{Extending and improving the algorithms.}
  The algorithm from Section~\ref{sec:partial} proves membership in $\x{\E}$ but can be asymptotically improved. 
  First, the algorithm ignores any potential symmetry between different inputs, which could be taken into account and improve the complexity by a factor of $k!$. This is in line with $\FPT$ algorithms of the form $f(k) \cdot |x|^{O(1)}$ where the function $f$ in a first attempt may be a tower of exponential, but which through clever algorithmic techniques in some cases can be improved to a single-exponential function $2^{ck}$ for reasonably small $c > 0$~\cite{DBLP:journals/siamdm/FominKLPS19}.
The algorithm in Section~\ref{sec:partial} and Section~\ref{sec:allen} can likely be improved to handle higher-arity constraints, although it would make the presentation and the analysis significantly more complex. 
To handle e.g. arity three constraints the number of sets to keep track of would almost double, as we would, for \ref{sec:partial}, also need to know how variables \emph{outside} of the temporary waists (those that in the algorithm are represented by $T_1^<,\dots,T_k^<$ and $T_1^>,\dots,T_k^>$) 
 relate to one of said waists.
 Another interesting direction is to consider different parameters: the algorithm from Section~\ref{sec:partial} can very likely be modify to use  the width of the partial order as the coarse-grained parameter instead of \ew, but it is less clear whether the dimension of a partial order is sufficient for an $\x{\E}$-algorithm.

\section*{Acknowledgements}
The first author is partially supported by the {\em National Graduate School in Computer Science} (CUGS), Sweden. The second author is partially supported by the Swedish Research Council (VR) under grant 2019-03690.

\bibliographystyle{plain}
\bibliography{main_arxiv}

\end{document}